\documentclass[conference,letterpaper,romanappendices]{ieeeconf}
\let\proof\relax   

\usepackage{amsthm,xpatch}
\usepackage{amsmath,amsfonts}
\usepackage{cite}
\usepackage{amssymb}
\usepackage{bm}
\usepackage{dsfont}
\usepackage{graphicx, subfigure}
\usepackage{color}
\usepackage{breqn}
\usepackage{mathtools}
\usepackage{bbm}
\usepackage{latexsym}
\usepackage[ruled, linesnumbered]{algorithm2e}
\usepackage{accents}
\usepackage{tikz}
\usepackage{soul}

\newtheorem{lemma}{Lemma}
\newtheorem{theorem}{Theorem}

\newtheorem{example}{Example}

\newcommand*{\transpose}{%
  {\mathpalette\@transpose{}}%
}

\newcommand\overmat[2]{%
  \makebox[0pt][l]{$\smash{\overbrace{\phantom{%
    \begin{matrix}#2\end{matrix}}}^{\text{$#1$}}}$}#2}

\IEEEoverridecommandlockouts

\begin{document}

\newcommand{\SB}[3]{
\sum_{#2 \in #1}\biggl|\overline{X}_{#2}\biggr| #3
\biggl|\bigcap_{#2 \notin #1}\overline{X}_{#2}\biggr|
}

\newcommand{\Mod}[1]{\ (\textup{mod}\ #1)}

\newcommand{\overbar}[1]{\mkern 0mu\overline{\mkern-0mu#1\mkern-8.5mu}\mkern 6mu}

\makeatletter
\newcommand*\nss[3]{%
  \begingroup
  \setbox0\hbox{$\m@th\scriptstyle\cramped{#2}$}%
  \setbox2\hbox{$\m@th\scriptstyle#3$}%
  \dimen@=\fontdimen8\textfont3
  \multiply\dimen@ by 4             
  \advance \dimen@ by \ht0
  \advance \dimen@ by -\fontdimen17\textfont2
  \@tempdima=\fontdimen5\textfont2  
  \multiply\@tempdima by 4
  \divide  \@tempdima by 5          
  \ifdim\dimen@<\@tempdima
    \ht0=0pt                        
    \@tempdima=\fontdimen5\textfont2
    \divide\@tempdima by 4          
    \advance \dimen@ by -\@tempdima 
    \ifdim\dimen@>0pt
      \@tempdima=\dp2
      \advance\@tempdima by \dimen@
      \dp2=\@tempdima
    \fi
  \fi
  #1_{\box0}^{\box2}%
  \endgroup
  }
\makeatother

\makeatletter
\renewenvironment{proof}[1][\proofname]{\par
  \pushQED{\qed}%
  \normalfont \topsep6\p@\@plus6\p@\relax
  \trivlist
  \item[\hskip\labelsep
        \itshape
    #1\@addpunct{:}]\ignorespaces
}{%
  \popQED\endtrivlist\@endpefalse
}
\makeatother

\makeatletter
\newsavebox\myboxA
\newsavebox\myboxB
\newlength\mylenA

\newcommand*\xoverline[2][0.75]{%
    \sbox{\myboxA}{$\m@th#2$}%
    \setbox\myboxB\null
    \ht\myboxB=\ht\myboxA%
    \dp\myboxB=\dp\myboxA%
    \wd\myboxB=#1\wd\myboxA
    \sbox\myboxB{$\m@th\overline{\copy\myboxB}$}
    \setlength\mylenA{\the\wd\myboxA}
    \addtolength\mylenA{-\the\wd\myboxB}%
    \ifdim\wd\myboxB<\wd\myboxA%
       \rlap{\hskip 0.5\mylenA\usebox\myboxB}{\usebox\myboxA}%
    \else
        \hskip -0.5\mylenA\rlap{\usebox\myboxA}{\hskip 0.5\mylenA\usebox\myboxB}%
    \fi}
\makeatother

\xpatchcmd{\proof}{\hskip\labelsep}{\hskip3.75\labelsep}{}{}

\pagestyle{plain}

\title{\fontsize{20}{28}\selectfont Single-Server Single-Message Online Private Information Retrieval with Side Information}


\author{Fatemeh Kazemi, Esmaeil Karimi, Anoosheh Heidarzadeh, and Alex Sprintson\thanks{The authors are with the Department of Electrical and Computer Engineering, Texas A\&M University, College Station, TX 77843 USA (E-mail: \{fatemeh.kazemi, esmaeil.karimi, anoosheh, spalex\}@tamu.edu).}}

\maketitle 

\thispagestyle{plain}

\begin{abstract}
In many practical settings, the user needs to retrieve information from a server in a periodic manner, over multiple rounds of communication. In this paper, we discuss the setting in which this information needs to be retrieved privately, such that the identity of all the information retrieved until the current round is protected. This setting can occur in practical situations in which the user needs to retrieve items from the server or a periodic basis, such that the privacy needs to be guaranteed for all the items been retrieved until the current round. We refer to this setting as an \emph{online private information retrieval} as the user does not know the identities of the future items that need to be retrieved from the server. 

Following the previous line of work by Kadhe \emph{et al.}~we assume that the user knows a random subset of $M$ messages in the database as a side information which are unknown to the server. Focusing on scalar-linear settings, we characterize the \emph{per-round capacity}, i.e., the maximum achievable download rate at each round, and present a coding scheme that achieves this capacity. The key idea of our scheme is to utilize the data downloaded during the current round as a side information for the subsequent rounds. We show for the setting with $K$ messages stored at the server, the per-round capacity of the scalar-linear setting is $C_1= ({M+1})/{K}$ for round $i=1$ and ${C_i= {(2^{i-1}(M+1))}/{KM}}$ for round $i\geq2$, provided that ${K}/({M+1})$ is a power of $2$.




\end{abstract}

\section{Introduction}

The goal of the Private Information Retrieval (PIR) schemes~\cite{chor1995private} is to enable a user to privately download, with minimum cost, a message belonging to a database with copies stored on a single or multiple remote servers, without revealing which message it is requesting. In a single server scenario, the entire database needs to be downloaded to preserve the privacy of the requested message. However, when the user has some side information about the database \cite{kadhe,heidarzadeh2018capacity,tandon2017capacity,wei2017fundamental,wei2018cache,chen2017capacity,shariatpanahi2018multi}, the information-theoretic privacy can be achieved more efficiently than downloading the whole database. 
 
In the PIR with side information setting, the user has access to a random subset of the messages in the database as a side information, which are unknown to the server. This side information could have been obtained from other trusted users or through previous interactions with the server. In this setting, the savings in the download cost depend on whether the user wants to protect only the privacy of the requested message, or the privacy of both the requested message and the messages in the side information. 
 
To the best of our knowledge, all of the prior works on the private information retrieval focus on the \emph{single} message request.  However, in many practical settings, the user needs to retrieve multiple messages periodically, over a period of time. For example, a user might  retrieve a news item or the stock market information on a daily basis. The key requirement in such scenarios is to protect the identity of all the requested messages up to the current round. 
By leveraging previously downloaded messages, the user can significantly reduce the number of bits that need to be downloaded. Accordingly, in this paper we analyze both the fundamental limits as well as the achievability schemes for the multi-round PIR schemes. We refer to this setting as an \emph{online private information retrieval} as the user does not know the identities of the future items that need to be retrieved from the server. 



\subsection{Main Contributions}

In this paper, we study the problem of single-server online PIR with side information. In this problem, a user wishes to download a sequence of messages ${\mathcal{X}_{W}=\{X_{W_1}, X_{W_2}, \cdots, X_{W_t}\}}$ from a database $\mathcal{X}$ of $K$ messages, stored on a single server. The communication is performed in rounds, such that at round $i$, the user wishes to retrieve a message $X_{W_i}$ for some $W_i\in [K]$. We assume that the user decides on which message $W_i$ to request at round $i$ and that the identity of the future messages  $W_j$, $j>i$ are not known at round $i$. We also assume that the user has access to $M$ messages which are selected uniformly at random and whose identity are not known to the server. 

We focus on the scenario where at round $i$, the user wishes to protect the privacy of all the requested messages up to round $i$, $\{W_1,\cdots,W_i\}$ for $1\leq i\leq t$. That is, after the user makes a request to the server at round $i$, the server cannot decide which of the $K$ messages is more likely to get requested at that round and at the previous rounds. Focusing on scalar-linear settings, we characterize the \emph{per-round capacity}, i.e., the maximum achievable download rate at each round. We also present a scalar-linear coding scheme that achieves this capacity. The key idea of our scheme is to utilize the data downloaded during the current round as a side information for subsequent rounds. We show for the setting with $K$ messages stored at the server and a random subset of $M$ messages available to the user at the first round, the per-round capacity of the scalar-linear scheme is ${C_1= ({M+1})/{K}}$ for round $i=1$ and ${{C_i= {(2^{i-1}(M+1))}/{KM}}}$ for round $i\geq 2$, provided that  ${K}/({M+1})=2^l$ for some ${l \geq 1}$. 
 

\subsection{Related Work}

The classical PIR problem with multiple servers each of which stores the full copy of the database, has been extensively studied \cite{sun2016capacity,sun2018capacity,banawan2018capacity}.  The most relevant to our paper is the line of work
that focuses on setting with multiple retrieved messages\cite{banawan2018multi,shariatpanahi2018multi,heidarzadeh} as well as settings in which the user access to certain files as side information before the information retrieval process begins. The side information settings have been studied in \cite{kadhe,heidarzadeh2018capacity} for the single server setting and in \cite{tandon2017capacity,wei2018cache,wei2017fundamental,chen2017capacity,shariatpanahi2018multi} for the multi-server setting. 

Kadhe \emph{et al.}~\cite{kadhe} have initiated the study of the single-server single-message PIR with side information. The multi-server single-message and the multi-server multi-message PIR with side information problems, in the scenario where the user wants to protect the privacy of both the requested message(s) and the messages in the side information, are studied in \cite{chen2017capacity} and \cite{shariatpanahi2018multi}, respectively. Heidarzadeh \emph{et al.}~\cite{heidarzadeh2018capacity} focused on a setting in which the side information is random linear combination of a subset of messages. 

To the best our knowledge, none of the prior works on the private information retrieval focused on  the \emph{online} settings in which the requests are issued one at a time such that the identity of future requests are unknown.

\section{Problem Formulation and Main Results}\label{sec:SN}



For a positive integer $i$, denote $[i]\triangleq\{1,\dots,i\}$. Let $\mathbb{F}_q$ be a finite field of size $q$ and $\mathbb{F}_{q^m}$ be an extension field of $\mathbb{F}_q$ of size $q^m$ for some prime $q\geq 2$ and $m\geq 1$. We assume that there is a server storing a set $\mathcal{X}$ of $K$ messages, $\mathcal{X}\triangleq \{X_1,\dots,X_K\}$, with each message $X_i$ being independently and uniformly distributed over $\mathbb{F}_{q^m}$, i.e., ${H(X_1) = \dots = H(X_K) = L}$ and ${H(X_1,\dots,X_K) = KL}$, where $L \triangleq m\log_2 q$. We assume that there is a user that wishes to retrieve a sequence of messages ${\mathcal{X}_{W}=\{X_{W_1}, X_{W_2}, \cdots,X_{W_t}\}}$ from the server so that at round $i$, the user wishes to retrieve the message $X_{W_i}$ for some $W_i\in [K]$. We assume that the identity of  the index $W_i$ of the message retrieved at round $i$ is not known to the user before round $i$.

We also assume that initially the user knows a random subset $\mathcal{X}_{\mathcal{S}} $ of $\mathcal{X}$ that includes $M$ messages for some $\mathcal{S} \subset [K]$, $|\mathcal{S}|=M$. We refer to $W_i$ as the \emph{demand index at round $i$}, $X_{W_i}$ as the \emph{demand at round $i$}, and refer to $\mathcal{S}$ as the \emph{side information index set}, $\mathcal{X}_{\mathcal{S}} $ as the \emph{side information} and $M$ as the  \emph{size of the side information set}. 

Let $\boldsymbol{S}$ and $\boldsymbol{W}_i$ be random variables corresponding to $\mathcal{S}$ and $W_i$, respectively. Denote the probability mass function (pmf) of $\boldsymbol{S}$ by $p_{\boldsymbol{S}}(\cdot)$, and the conditional pmf of $\boldsymbol{W}_i$ given $\boldsymbol{S}$ by $p_{\boldsymbol{W}_i|\boldsymbol{S}}(\cdot|\cdot)$. We assume that $\mathcal{S}$ is uniformly distributed over all subsets of $[K]$ of size $M$, i.e.,

\[p_{\boldsymbol{S}}(\mathcal{S}) = \frac{1}{\binom{K}{M}},\quad \mathcal{S} \subset [K], |\mathcal{S}|=M\] and $\boldsymbol{W}_i$'s are independent and uniformly distributed over $[K]\setminus \mathcal{S}$, i.e., 

\begin{equation*}
p_{\boldsymbol{W}_i|\boldsymbol{S}}(W_i|\mathcal{S}) = 
\left\{\begin{array}{ll}
\frac{1}{K-M}, & W_i \notin \mathcal{S}\\	
0, & \text{otherwise}.
\end{array}\right.	
\end{equation*}

Also, we assume that the server knows the size of $\mathcal{S}$ (i.e., $M$), the pmf $p_{\boldsymbol{S}}(.)$ and $p_{\boldsymbol{W}_i|\boldsymbol{S}}(.|.)$, but the realizations $\mathcal{S}$ and $W_i$ are unknown to the server a priori. 

At round $i$, for any $\mathcal{S}$ and $W_i$, in order to retrieve $X_{W_i}$, the user sends to the server a query $Q^{[W_i,\mathcal{S}]}$, and upon receiving ${Q}^{[W_i,\mathcal{S}]}$, the server sends to the user an answer ${A}^{[W_i,\mathcal{S}]}$. We define $Q^{[W_{1:i},\mathcal{S}]}\triangleq \{Q^{[W_1,\mathcal{S}]},Q^{[W_2,\mathcal{S}]},\cdots,Q^{[W_i,\mathcal{S}]}\}$ and $A^{[W_{1:i},\mathcal{S}]}\triangleq \{A^{[W_1,\mathcal{S}]},A^{[W_2,\mathcal{S}]},\cdots,A^{[W_i,\mathcal{S}]}\}$ as the set of all queries and answers up to the round $i$, respectively.

 Note that the query at round $i$, $Q^{[W_i,\mathcal{S}]}$ is a (potentially stochastic) function of $W_i, \mathcal{S}, X_{\mathcal{S}}, Q^{[W_{1:{i-1}},\mathcal{S}]}$, and  $A^{[W_{1:{i-1}},\mathcal{S}]}$. Similarly, the answer at round $i$, $A^{[W_i,\mathcal{S}]}$ is a (deterministic) function of ${Q}^{[W_{1:i},\mathcal{S}]}$ and the messages in $\mathcal{X}$, i.e., \[H({A}^{[W_i,\mathcal{S}]}| {Q}^{[W_{1:i},\mathcal{S}]},\mathcal{X}) = 0.\]

The queries from the first round up to round $i$ all together, $Q^{[W_{1:i},\mathcal{S}]}$ must protect the \emph{individual privacy} 
 for all the user's demand indices up to round $i$ from the server, i.e., \[\mathbb{P}(\boldsymbol{W}_j= W'| Q^{[W_{1:i},\mathcal{S}]}, {\mathcal{X}})= \frac{1}{K}\]for all $W' \in [K]$ and all $j \in [i]$. This means that some correlations between the demands in different rounds can be revealed to the server, but all the demands up to round $i$ must be protected to be individually private at each round. This condition is referred to as the \emph{privacy condition}. 

All the answers from the first round up to round $i$, $A^{[W_{1:i},\mathcal{S}]}$ along with the side information ${\mathcal{X}}_{\mathcal{S}}$ must enable the user to retrieve the demand $X_{W_i}$. This condition is referred to as the \emph{recoverability condition}, as follows: \[H(X_{W_i}| A^{[W_{1:i},\mathcal{S}]}, Q^{[W_{1:i},\mathcal{S}]}, {\mathcal{X}}_{\mathcal{S}})=0.\]

The problem of the single-server \emph{Online Private Information Retrieval (OPIR)} is to design a query $Q^{[W_i,{\mathcal{S}}]}$ and an answer $A^{[W_i,{\mathcal{S}}]}$ for any given ${\mathcal{S}}$ and $W_i$ at round $i \geq 1$, that satisfy the privacy and recoverability conditions. 

The \emph{per-round rate} of an OPIR algoirhtm at round $i$ denoted by $R_i$, is defined as the ratio of the entropy of a message, i.e., $L$, to the maximum entropy of the answer at round $i$, i.e., $$\displaystyle R_i=\min_{W_{1:i},\mathcal{S}}\frac{L}{H(A^{[{W_i},\mathcal{S}]})}.$$  The \emph{per-round capacity} of \emph{OPIR} at round $i$ denoted by $C_{\text{\it i}}$, is defined as the supremum of rates over all \emph{OPIR} algorithms that achieve the capacity up to round $i-1$. We focus on scalar-linear capacity, which corresponds to the maximum rate that can be achieved by scalar-linear protocols.

The goal of this paper is to establish the scalar-linear per-round capability of the OPIR and present an algorithm that achieves this capacity. Theorem \ref{Capacity 1} characterizes the capacity of scalar-linear OPIR problem for the case when $\frac{K}{M+1}=2^l$ for some $l \geq 1$. 

\begin{theorem}\label{Capacity theorem}
For the OPIR problem with $K$ messages, and side information of size $M$, when ${K}/({M+1})=2^l$ for some $l \geq 1$, the scalar-linear per-round capacity at round $i$ is given by:
\begin{dmath}\label{Capacity 1}
C_i=
\begin{cases}
{\frac{M+1}{K}} & \quad i=1\\
\frac{2^{i-1}(M+1)}{KM} & \quad i \geq 2\\
\end{cases}
\end{dmath}
\end{theorem}\vspace{0.05cm}

%
%
%

\section{Converse Proof}\label{converse proof}

In this section, we prove the converse part of Theorem~\ref{Capacity theorem}.
%
%
%
%
%
%
Suppose that the user wishes to retrieve a sequence of messages ${\mathcal{X}}_{W}=\{{{X}}_{W_1}, {{X}}_{W_2}, \cdots, {{X}}_{W_t}\}$ from the server so that at round $i$, the user wants to download the message ${{X}}_{W_i}$ for some $W_i\in [K]$, and knows ${\mathcal{X}}_{\mathcal{S}}$ for a given ${\mathcal{S}} \subseteq [K] \setminus W$, $|{\mathcal{S}}|= M$. By assumption, ${K}/({M+1})=2^l$ for some $l \geq 1$. At round $i$, for any ${\mathcal{S}}$ and $W_i$, in order to retrieve ${{X}}_{W_i}$, the user sends to the server a query $Q^{[W_i,{\mathcal{S}}]}$, and the server responds to the user by an answer ${A}^{[W_i,{\mathcal{S}}]}$. The answer ${A}^{[W_i,{\mathcal{S}}]}$ at round $i$ is a set of $m_i$ messages, i.e., ${{A}^{[W_i,{\mathcal{S}}]}\triangleq \{y_{i,1},\dots,y_{i,{m_i}}\}}$. In linear OPIR schemes each message $y_{i,j}$ for $1\leq j \leq {m_i}$ is a linear combination of the original messages in ${\mathcal{X}}$, i.e. \[y_{i,j}=\sum_{m \in [K]}\gamma_{i,j}^m {{X}}_m,\]
where $\gamma_{i,j}^m \in \mathbb{F}_q$ are the encoding coefficients of $y_{i,j}$. We refer to the vector $\gamma_{i,j}=[\gamma_{i,j}^1,\gamma_{i,j}^2,\cdots,\gamma_{i,j}^K]$ as the encoding vector of $y_{i,j}$. The $i$-th unit encoding vector that corresponds to the original packet ${{X}}_i$ is denoted by $u_i=[u_i^1,u_i^2,\cdots,u_i^K]$, where $u_i^i=1$ and $u_i^j=0$ for $i \neq j$. Consider the set of $K$ linearly independent unit vectors $\{u_1,u_2,\cdots,u_K\}$ as a basis for a vector space $\mathcal{V}$ of dimension $K$. Thus, the encoding vector of $y_{i,j}$, i.e., $\gamma_{i,j}$, is a vector in the vector space $\mathcal{V}$. Define the answer matrix at round $i$, $A_i$, of dimension ${({m_i} \times K)}$ with $\gamma_{ij}$ being the $j$-th row of $A_i$. Note that the number of messages in ${A}^{[W_i,{\mathcal{S}}]}$, or equivalently, the number of rows of matrix $A_i$ is the download cost at round $i$.  

For the first round ($i=1$), the proof of converse follows from the prior results for PIR with side information (see \cite[Lemma~1]{kadhe}). It is easy to verify that at round $1$, any optimal scalar-linear scheme can be converted to the partition-based scheme of \cite{kadhe} by row operations. The answer matrix $A_1$ corresponding to the optimal scheme has exactly $k/(M+1)$ rows. Followed by a column permutation, the matrix $A_1$ can be represented as:\vspace{0.2cm}

\[
  A_1= \begin{array}{@{} c @{}}
    \left [
      \begin{array}{ *{21}{c} }
        \overmat{M+1}{\star & \cdots & \star}\\
         & & & \overmat{M+1}{\star & \cdots & \star}\\
         & & & & & & \ddots\\
         & & & & & &  & \overmat{M+1}{\star & \cdots & \star}
      \end{array}
    \right ]
  \end{array}_{\frac{K}{M+1}\times K}
\] where $\star$'s indicate non-zero entries in matrix $A_1$ and all other entries in matrix $A_1$ are zero. Each row of $A_1$ corresponds to one of the messages in the answer. For instance, the first row corresponds to ${X_1+\cdots+X_{M+1}}$, the second row corresponds to ${X_{M+2}+\cdots+X_{2M+2}}$, and so on. The support set of each message in the answer is called a \emph{partition}. Thus, the optimal scheme in the first round has $n=K/(M+1)$ number of partitions, denoted by $\{P_1, P_2,\cdots,P_n\}$.

For the second round and after that ($i \geq 2$), in Theorem~\ref{theorem2} we prove that the maximum entropy of the answer, i.e.,  $H(A^{[{W_i},\mathcal{S}]})$, where the maximum is taken over all $W_i$ and ${\mathcal{S}}$, is lower bounded by ${KM}/({2^{i-1}(M+1)})$. 

\begin{theorem}\label{theorem2}

The maximum entropy of the answer $H(A^{[{W_i},\mathcal{S}]})$ at round $i \geq 2$ over all $W_i$ and ${\mathcal{S}}$, is lower bounded by ${KM}/({2^{i-1}(M+1)})$. 
\end{theorem}

\begin{proof}
For linear schemes it is sufficient to prove that the maximum number of rows of matrix $A_i$ for $i \geq 2$ 
is lower bounded by ${KM}/({2^{i-1}(M+1)})$. The proof is based on an inductive argument and uses a simple yet powerful observation, formally stated in Lemma \ref{privac & recov Conds 1}.


\begin{lemma}\label{privac & recov Conds 1}
For any $i$, $W_{1:i}$, ${\mathcal{S}}$, and any $W^{\star}\in [K]$, there must exist ${\mathcal{S}}^\star \subseteq P_j$, $|{\mathcal{S}}^\star|=M$ and ${W^{\star} \nsubseteq P_j}$ for some ${j \in [n]}$ such that \[H({{X}}_{W^{\star}}| A^{[W_{1:i},{\mathcal{S}}]}, Q^{[W_{1:i},{\mathcal{S}}]}, {\mathcal{X}}_{{\mathcal{S}}^\star}) = 0.\]
\end{lemma}

\begin{proof}
The proof is given in Appendix.
\end{proof}

\begin{lemma}\label{privac & recov Conds 2}
At round $i$ for $i \geq 2$, in the vector space spanned by the rows of matrices $A_1, A_2,\cdots, A_i$, corresponding to any $W^\star \in [K]$, there must exist a vector which is a linear combination of at most $M+1$ messages including ${\mathcal{X}}_{W^\star}$ itself and at most $M$ other messages which are a subset of ${\mathcal{X}}_{{\mathcal{S}}^\star}$, a potential side information for ${\mathcal{X}}_{W^\star}$ defined in Lemma~\ref{privac & recov Conds 1}.
\end{lemma}
\begin{proof}
The proof is based on contradiction. Assume that at round $i$, for $i \geq 2$, in the vector space spanned by the rows of matrices $A_1, A_2,\cdots, A_i$, for a given $W^\star \in [K]$, there does not exist such a vector described in Lemma~\ref{privac & recov Conds 2}. This means that ${{X}}_{W^{\star}}$ is not recoverable from $A^{[W_{1:i},{\mathcal{S}}]}$ and ${\mathcal{X}}_{{\mathcal{S}}^\star}$, which contradicts the result of Lemma~\ref{privac & recov Conds 1}.
\end{proof}


 In fact, in the vector space spanned by the rows of $A_1, A_2,\cdots, A_i$, there must exist $K$ of such vectors, one for each potential value of $W^{\star} \in [K]$. Define matrix $\Gamma$ with these $K$ vectors being as the rows of $\Gamma$. An instance of matrix $\Gamma$ would be as follows:\vspace{0.25cm}

\[
  {\Gamma}= \begin{array}{@{} c @{}}
        \begin{array}{@{} r @{}}
      \text{R}~\{\hspace{\nulldelimiterspace} \\
      \\
      \\
      \\
      \\
      \\
      \text{L}~\left\{\begin{array}{@{}c@{}}\end{array}\right.
    \end{array}
    \left [
      \begin{array}{ *{5}{c} }
        \overmat{W^\star}{1} & & \overmat{{\mathcal{S}}^\star}{\star \star \cdots \star}\\
        \\
         & 1 & & {\star \star \cdots \star}\\
          & & \ddots\\
         & {\star \star \cdots \star} & & 1\\
         \\
         & & {\star \star \cdots \star} & &  1
      \end{array}
    \right ]
  \end{array}_{K\times K}
\]\vspace{0.1cm}

\begin{lemma}\label{privac & recov Conds 3}
The rank of matrix $\Gamma$ is lower bounded by $K/2$. 
\end{lemma}

\begin{proof}

 Since by Lemma~\ref{privac & recov Conds 1}, ${\mathcal{S}}^\star \subseteq P_j$ for some ${j \in [n]}$ and $W^{\star} \nsubseteq P_j$, then ${\mathcal{S}}^\star$ is either in the left side of $W^\star$, or in the right side of $W^\star$ in each row of matrix $\Gamma$. Accordingly, the rows of matrix $\Gamma$ can be classified into two types: L and R, based on the criteria that ${\mathcal{S}}^\star$ is in the left side of $W^\star$, or in the right side of $W^\star$, respectively. Let $z_1$ and $z_2$ denote the number of rows of type L and the number of rows of type R, respectively. It is easy to verify that the maximum of $z_1$ and $z_2$ is greater than or equal to $K/2$, i.e., $\max(z_1, z_2)\geq K/2$. Without loss of generality, assume $\max(z_1, z_2)=z_1$. Then, $z_1 \geq K/2$. By removing $z_2$ rows of type R from matrix $\Gamma$, we are left with $z_1 \geq K/2$ rows of type L that constitute a matrix of size $z_1 \times K$, in which there exists a lower triangular submatrix of size $z_1 \times z_1$ and rank $z_1 \geq K/2$. Thus, the rank of matrix $\Gamma$ is at least $z_1$ which is lower bounded by $K/2$.
\end{proof}

For the second round ($i=2$), we need to show that the number of rows of matrix $A_2$ is lower bounded by ${KM}/({2(M+1)})$. Based on Lemma~\ref{privac & recov Conds 2}, in the vector space spanned by the rows of matrices $A_1$ and $A_2$, there must exist all $K$ rows of matrix $\Gamma$ which based on Lemma~\ref{privac & recov Conds 3} is of rank greater than or equal to $K/2$. On the other hand, as mentioned earlier, the optimal scheme in the first round is partitioning where each row of $A_1$ corresponds to a linear combination of $M+1$ messages. One can readily confirm that corresponding to any $M+1$ number of linearly independent rows of matrix $\Gamma$, there exists at most one linear combination of these rows in the rows of matrix $A_1$. Thus, there must exist at least $M$ linearly independent combinations of these rows in the rows of matrix $A_2$. Then, we have: \[\mathrm{rank} (A_2) \geq \frac{M}{M+1} \times \mathrm{rank}  ({\mathcal{S}}) \geq \frac{M}{M+1} \times \frac{K}{2} = \frac{KM}{2(M+1)}\] 

In other words, matrix $\Gamma$ has at least $K/2$ number of linearly independent rows. Thus, there exist at most $K/(2(M+1))$ linearly independent combinations of these rows in the rows of matrix $A_1$. Therefore, there must exist at least ${K/2-K/(2(M+1))=KM/(2(M+1))}$ linearly independent combinations of these rows in the rows of matrix $A_2$, which indicates that the number of rows of matrix $A_2$ is lower bounded by ${KM}/({2(M+1)})$. The optimal scheme achieves the lower bound. Thus, in the optimal scheme, the number of rows of matrix $A_2$ is exactly ${KM}/({2(M+1)})$.
  
For the third round ($i=3$), we need to show that the number of rows of matrix $A_3$ is lower bounded by ${KM}/({4(M+1)})$. Based on the Lemma~\ref{privac & recov Conds 2}, in the vector space spanned by the rows of matrices $A_1$, $A_2$ and $A_3$, there must exist all $K$ rows of matrix $\Gamma$ which based on Lemma~\ref{privac & recov Conds 3} is of rank greater than or equal to $K/2$. By the same reasoning as in the case of $i=2$, corresponding to any $2(M+1)$ number of linearly independent rows of matrix $\Gamma$, there exist at most two linearly independent combinations of these rows in the rows of matrix $A_1$. On the other hand, we showed that $A_2$ in the optimal scheme is of rank ${KM}/({2(M+1)})$, which shows that corresponding to any $2(M+1)$ number of linearly independent rows of matrix $\Gamma$, there exist at most $M$ linearly independent combinations of these rows in the rows of matrix $A_2$. Thus, there must exist at least $2(M+1)-2-M=M$ linearly independent combinations of these rows in the rows of matrix $A_3$. Then, we have: 
\[\mathrm{rank}  (A_3) \geq \frac{M}{2(M+1)} \times \mathrm{rank}  (S) \geq  \frac{KM}{4(M+1)}\] 

In other words, matrix $\Gamma$ has at least $K/2$ number of linearly independent rows. Thus, there exist at most $K/(2(M+1))$ linearly independent combinations of these rows in the rows of matrix $A_1$, and ${KM/(4(M+1))}$ linearly independent combinations of these rows in the rows of matrix $A_2$. Therefore, there must exist at least ${{K}/{2}-{K}/{(2(M+1))}-{KM}/{(4(M+1))}}={KM}/{(4(M+1))}$ linearly independent combinations of these rows in the rows of matrix $A_3$, which indicates that the number of rows of matrix $A_3$ is lower bounded by ${KM}/({4(M+1)})$.
 
Using the same proof technique and similar reasoning as in the cases of $i=2$ and $i=3$, it can be shown that the number of rows of $A_i$ for $i \geq 2$ is lower bounded by $KM/(2^{i-1}(M+1))$. By the result of Lemma~\ref{privac & recov Conds 2}, in the vector space spanned by the rows of matrices $A_1, A_2, \cdots, A_i$, there must exist all $K$ rows of matrix $\Gamma$, which is of rank greater than or equal to $K/2$ (by Lemma~\ref{privac & recov Conds 3}). Again, similarly as in the cases of $i=2$ and $i=3$, it follows that corresponding to any $2^{i-2}(M+1)$ number of linearly independent rows of matrix $\Gamma$, there exist at most $2^{i-2}, 2^{i-3}M, 2^{i-4}M,\cdots, M$, linearly independent combinations of these rows in the rows of matrix  $A_1, A_2, A_3, \cdots, A_{i-1}$, respectively. Thus, there must exist at least ${2^{i-2}(M+1)-(2^{i-2})-(\sum_{j=0}^{i-3} 2^j)M=M}$ linearly independent combinations of these rows in the rows of matrix $A_i$. Then, we have:
\[\mathrm{rank} (A_i) \geq \frac{M}{2^{i-2}(M+1)} \times \frac{K}{2}= \frac{KM}{2^{i-1}(M+1)}\] which shows that the number of rows of matrix $A_i$ is lower bounded by ${KM}/({2^{i-1}(M+1)})$. 
\end{proof}

\section{Achievability Scheme}\label{achievability Scheme}


In this section, we propose an OPIR protocol for arbitrary $K$ and $M$ where $K/(M+1)=2^l$ for some ${l \geq 1}$, which achieves the rate ${(M+1)}/{K}$ in the first round and ${(2^{i-1}(M+1))}/{KM}$ at rounds $i \geq 2$. The proposed scheme, referred to as the \textit{Online Partitioning (OP) Protocol}, is described in the following. \vspace{0.2cm} 

Each round of the OP protocol consists of four steps:\vspace{0.2cm}

\noindent {\textbf{Round $\bm{i=1}$}:}\vspace{0.2cm} 

\textit{\textbf{Step 1:}} The user creates a partition of the $K$ messages into $n_1\triangleq K/(M+1)$ sets. The first partition, $P_1^1$ is formed by combining the demand and the side information set $\mathcal{S}$: $P_1^1 \triangleq \{W_1\} \cup \mathcal{S}$. The user randomly partitions the set of messages $[K]\setminus P_1^1$ into $n_1-1$ sets, each of size $M + 1$, denoted as $P_2^1,\cdots,P_{n_1}^{1}$. 

\textit{\textbf{Step 2:}} The user sends to the server a uniform random permutation of the partition ${\{P_1^1,\cdots,P_{n_1}^{1}\}}$, i.e., it sends
${\{P_1^1,\cdots,P_{n_1}^{1}\}}$ in a random order. 

\textit{\textbf{Step 3:}}  The server picks an arbitrary Cauchy matrix of size ${K \times (Ml+1)}$ denoted by ${C}\triangleq [c_{i,j}]$ with parameters from $\mathbb{F}_q$, where ${q \geq K+Ml+1}$. For a subset $\mathcal{P} \subset [K]$, let $\mathcal{V_P}$ denote the characteristic vector of the set $\mathcal{P}$, which is a vector of length $K$ such that for all $i \in [K]$, its $i$-th entry is $c_{i,1}$ if $i \in \mathcal{P}$, otherwise it is $0$. The server computes the answer $A^{[W_{1},S]}$ as a set of $n_1$ inner products given by ${A^{[W_{1},S]}= \{A_{P_{1}^{1}},\cdots,A_{P_{n_1}^{1}}\}}$, where ${A_{\mathcal{P}}=[X_1,\cdots,X_K] \cdot \mathcal{V_P}}$ for all ${\mathcal{P} \in \{P_1^{1},P_2^{1}, \cdots, P_{n_1}^{1}\}}$.  

\textit{\textbf{Step 4:}} Upon receiving the answer from the server, the user decodes $X_{W_1}$ by subtracting off the contributions of its side information ${\mathcal{X}}_{\mathcal{S}}$ from $A_P$ for some ${P \in \{P_1^{1},P_2^{1}, \cdots, P_{n_1}^{1}\}}$ such that $W_1 \in P$.
\vspace{0.21 cm} 

\noindent {\textbf{Round $\bm{i \geq 2}$}:}\vspace{0.2cm} 

\textit{\textbf{Step 1:}} Based on the OP protocol at round $i-1$, the user sends to the server a uniform random
permutation of a partition of the $K$ messages into $n_{i-1}$ number of sets as the query, i.e.,  $Q^{[W_{i-1},{\mathcal{S}}]}=\{P_1^{i-1},P_2^{i-1}, \cdots, P_{n_{i-1}}^{i-1}\}$. Given the query at round $i-1$, i.e., $Q^{[W_{i-1},{\mathcal{S}}]}$ and given $W_i$ and ${\mathcal{S}}$, the user creates a partition of $K$ messages into $n_{i}=\frac{n_{i-1}}{2}$ sets $\{P_1^{i},P_2^{i}, \cdots, P_{n_i}^{i}\}$. It should be noted that $n_{i-1}$ is always divisible by 2 based on the assumption that $K/(M+1)$ is a power of 2. It is easy to confirm that $W_i$ and ${\mathcal{S}}$ belong to two different partitions at round $i-1$. Assume, without loss of generality, that ${W_i \in P_1^{i-1}}$ and ${\mathcal{S}} \in P_2^{i-1}$. Then, the user constructs ${P_1^{i}=P_1^{i-1} \cup P_2^{i-1}}$. For constructing each $P^{i} \in \{P_2^{i}, \cdots, P_{n_i}^{i}\}$, the user chooses two partitions from the remaining partitions at round $i-1$ uniformly at random (without replacement) and unions them.

\textit{\textbf{Step 2:}} The user sends to the server a uniform random
permutation of the partition $\{P_1^{i},P_2^{i}, \cdots, P_{n_i}^{i}\}$, i.e., it sends $Q^{[W_{i},{\mathcal{S}}]}=\{P_1^{i},P_2^{i}, \cdots, P_{n_i}^{i}\}$ in a random order.

\textit{\textbf{Step 3:}} The server extracts a submatrix ${{H}=[h_{ij}]_{K \times M}}$ of the Cauchy matrix ${C}$ which contains all rows and $j$-th column of the matrix $C$, where ${(i-2)M+2 \leq j \leq (i-1)M+1}$. For a subset $\mathcal{P} \subset [K]$, let $G_{\mathcal{P}}=[g_{ij}]_{K \times M}$ denotes the characteristic matrix of the set $\mathcal{P}$, which is a matrix of size $K \times M$ such that for each $j \in [M]$, for all $i \in [K]$, ${g_{ij}=h_{ij}}$ if $i \in \mathcal{P}$, otherwise it is zero, i.e.,  $g_{ij}=0$ for $i \notin \mathcal{P}$. The server computes $A_{\mathcal{P}}=[X_1,\cdots,X_K] \cdot G_{\mathcal{P}}$ for all ${\mathcal{P} \in \{P_1^{i},P_2^{i}, \cdots, P_{n_i}^{i}\}}$. In fact, ${A_{\mathcal{P}}=[(A_{{\mathcal{P}}})_1,\cdots,(A_{{\mathcal{P}}})_M]}$ is a row vector of length $M$ that gives $M$ linearly independent combinations of the messages $X_i$ for $i \in [\mathcal{P}]$. Finally, the server computes the answer $A^{[W_{i},{\mathcal{S}}]}$ as a set of ${n_i\times M}$ linearly independent combinations of the messages, i.e., $A^{[W_{i},{\mathcal{S}}]}= \{(A_{P_{1}^{i}})_1,\cdots,(A_{P_{1}^{i}})_M, \cdots, (A_{P_{n_i}^{i}})_1,\cdots,(A_{P_{n_i}^{i}})_M\}$.
 
\textit{\textbf{Step 4:}} Upon receiving the answer from the server, the user retrieves $X_{W_i}$ by subtracting off the contributions of its side information ${\mathcal{X}}_{\mathcal{S}}$ from $A^{[W_{1:i},{\mathcal{S}}]}$ and solving a set of ${2^{i-2}(M+1)}$ linearly independent equations with ${2^{i-2}(M+1)}$ unknowns. It should be noted that every submatrix of a Cauchy matrix is itself a Cauchy matrix, which guarantees the existence of ${2^{i-2}(M+1)}$ linearly independent equations.

\begin{lemma}\label{lemma 4}
The OP protocol satisfies the recoverability and individual privacy conditions, while achieving the rate ${{(M+1)}/{K}}$ at first round, and the rate $(2^{i-1}(M+1))/{KM}$ at round $i \geq 2$.
\end{lemma}

\begin{proof}
The OP protocol for the first round is based the Partition and Code PIR Scheme which satisfies the recoverability and the privacy conditions and achieves the rate ${(M+1)}/{K}$ \cite{kadhe}. It should be noted that in the first round, the coefficients of the messages in the answer which are chosen according to a Cauchy matrix, have no effect on the recoverability and the privacy proofs.

In the OP protocol at round $i \geq 2$, the answer $A^{[W_i,{\mathcal{S}}]}$ consists of $n_iM={KM}/({2^{i-1}(M+1)})$ linear combinations of the messages in $X$, i.e., ${A^{[W_{i},{\mathcal{S}}]}=\{(A_{P_{1}^{i}})_1,\cdots,(A_{P_{n_i}^{i}})_M\}}$. Since the messages in $X$ are uniformly and independently distributed over $\mathbb{F}_{q^m}$, and $\{(A_{P_{1}^{i}})_1,\cdots,(A_{P_{n_i}^{i}})_M\}$ are linearly independent combinations of the messages, then $\{(A_{P_{1}^{i}})_1,\cdots,(A_{P_{n_i}^{i}})_M\}$ are uniformly and independently distributed over $\mathbb{F}_{q^m}$, i.e., ${H((A_{P_{1}^{i}})_1) = \dots = H((A_{P_{n_i}^{i}})_M) = m\log_2 q=L}$, and ${H(A^{[W_{i},{\mathcal{S}}]})={H((A_{P_{1}^{i}})_1)}+\cdots+H((A_{P_{n_i}^{i}})_M)=n_iML}$. Therefore, the rate of the OP protocol at round $i \geq 2$ is equal to $L/H(A^{[W_{i},{\mathcal{S}}]})={(2^{i-1}(M+1))}/{KM}$.

From the step $4$ of the OP protocol for round $i \geq 2$, it can be easily verified that the recoverability condition is satisfied because of choosing the coefficients from a Cauchy matrix.

To prove that the OP protocol satisfies the privacy condition at round $i \geq 2$, we need to show that  ${\mathbb{P}(\boldsymbol{W}_j= W'| Q^{[W_{1:i},{\mathcal{S}}]}, X)= \frac{1}{K}}$ for all $W' \in [K]$ and all $j \in [i]$. Since the OP protocol does not depend on the contents of the messages in $X$, then it is sufficient to prove that $\mathbb{P}(\boldsymbol{W}_j= W'| Q^{[W_{1:i},{\mathcal{S}}]})= \frac{1}{K}$ for all $W' \in [K]$ and all $j \in [i]$. Thus, for $j=i$, we have:

\begin{align*}
&\mathbb{P}(\boldsymbol{W}_i= W'| Q^{[W_{1:i},{\mathcal{S}}]})
\\ &= \sum_{{\mathcal{S}}^\star} (\mathbb{P}(\boldsymbol{W}_i= W'| Q^{[W_{1:i},{\mathcal{S}}]}, {\mathcal{S}}^\star)  \times \mathbb{P}( {\mathcal{S}}^\star|Q^{[W_{1:i},{\mathcal{S}}]})
\end{align*}

where the sum is over all possible ${\mathcal{S}}^\star$ of size $M$, a potential side information for ${W_i}= W'$. Let assume $W'$ is located in the $k_{th}$ partition of round $i$, i.e. $P_k^{i}$. As mentioned earlier, at round $i \geq 2$, each partition is a union of two partitions at round $i-1$. Without loss of generality assume that, $k_{th}$ partition of round $i$ (of size $2^{i-1}(M+1)$), is a union of $w_{th}$ and $v_{th}$ partitions of round $i-1$ (each of size $2^{i-2}(M+1)$), i.e., $P_k^{i}=P_w^{i-1} \cup P_v^{i-1}$, and $W'$ is located in the $P_w^{i-1}$. A possible potential side information ${\mathcal{S}}^\star$ for ${W_i}= W'$, would be a subset of $P_v^{i-1}$ of size $M$ which is completely located in one of the partitions of the first round. $P_v^{i-1}$ of size $2^{i-2}(M+1)$, is a union of $2^{i-2}$ partitions of the first round. From each partition of the first round, all subsets of size $M$, i.e., $\binom{M+1}{M}$, can be considered as possible ${\mathcal{S}}^\star$. Thus, one can consider $2^{i-2}\binom{M+1}{M}$ number of possible potential side information ${\mathcal{S}}^\star$ for ${W_i}= W'$. For a specific ${\mathcal{S}}^\star=\hat{{\mathcal{S}}}$, given that ${\mathcal{S}}^\star=\hat{{\mathcal{S}}}$ belongs to $P_v^{i-1}$, each of the elements in $P_w^{i-1}$ (of size $2^{i-2}(M+1)$) can be user's demand index with the same probability. In other words, ${\mathbb{P}(\boldsymbol{W}_i= W'| Q^{[W_{1:i},{\mathcal{S}}]},{\mathcal{S}}^\star=\hat{\mathcal{S}})=\frac{1}{2^{i-2}(M+1)}}$.

The probability of a specific ${\mathcal{S}}^\star=\hat{{\mathcal{S}}}$ given $Q^{[W_{1:i},{\mathcal{S}}]}$, can be computed by the application of the total probability theorem and chain rule, we have:\vspace{-0.4cm}

\begin{align*}
&\mathbb{P}({\mathcal{S}}^\star=\hat{\mathcal{S}}|Q^{[W_{1:i},{\mathcal{S}}]})\\ &=\sum_{j \in [\frac{K}{M+1}]}\mathbb{P}( {\mathcal{S}}^\star=\hat{{\mathcal{S}}}|Q^{[W_{1:i},{\mathcal{S}}]}, {\mathcal{S}}^\star \in P_j^{1}) \times \mathbb{P}( {\mathcal{S}}^\star \in P_j^{1}|Q^{[W_{1:i},{\mathcal{S}}]})
\end{align*}

where the potential side information ${\mathcal{S}}^\star$ belongs to each of the partitions of the first round with the same probability, i.e., $\mathbb{P}( {\mathcal{S}}^\star \in P_j^{1}|Q^{[W_{1:i},{\mathcal{S}}]})=\frac{M+1}{K}$ for all ${j \in [\frac{K}{M+1}]}$. In each partition of the first round, there are $\binom{M+1}{M}$ subsets of size $M$, each of which can be the potential side information  ${\mathcal{S}}^\star$ with the same probability. Thus, ${\mathbb{P}( {\mathcal{S}}^\star=\hat{{\mathcal{S}}}|Q^{[W_{1:i},{\mathcal{S}}]}, {\mathcal{S}}^\star \in P_j^{1})=\frac{1}{M+1}}$, for one $j \in [\frac{K}{M+1}]$ and zero for others. Thus, we have:
\begin{align*}
&\mathbb{P}({\mathcal{S}}^\star|Q^{[W_{1:i},{\mathcal{S}}]})= \frac{1}{M+1} \times \frac{M+1}{K}
\end{align*}

 Finally, $\mathbb{P}(\boldsymbol{W}_i= W'| Q^{[W_{1:i},{\mathcal{S}}]})$ is obtained as follows:\vspace{-0.4cm}

\begin{align*}
&\mathbb{P}(\boldsymbol{W}_i= W'| Q^{[W_{1:i},{\mathcal{S}}]})
\\ &= \sum_{{\mathcal{S}}^\star} (\mathbb{P}(\boldsymbol{W}_i= W'| Q^{[W_{1:i},{\mathcal{S}}]}, {\mathcal{S}}^\star)  \times \mathbb{P}( {\mathcal{S}}^\star|Q^{[W_{1:i},{\mathcal{S}}]})
\\ &=2^{i-2} (M+1) \times \frac{1}{2^{i-2}(M+1)} \times  \frac{1}{M+1}\times \frac{M+1}{K}=\frac{1}{K}.
\end{align*}

We proved that $\mathbb{P}(\boldsymbol{W}_i= W'| Q^{[W_{1:i},{\mathcal{S}}]})= \frac{1}{K}$ for all ${W' \in [K]}$. Using the same proof technique, for $j \in [i-1]$, it can be shown that $\mathbb{P}(\boldsymbol{W}_j= W'| Q^{[W_{1:i},{\mathcal{S}}]})= \frac{1}{K}$ for all $W' \in [K]$.

\end{proof}

\begin{example} (OP Protocol) Assume that the server has ${K=12}$ messages $\{X_1, X_2, \cdots, X_{12}\}$, and the user has $M=2$ messages, $X_2$ and $X_3$, as side information, i.e., ${\mathcal{S}}=\{2,3\}$. Consider a scenario as follows:\vspace{0.1cm}

\textbf{First round:} The user demands the message $X_1$, i.e., $W_1=1$. Thus, the user labels $4$ sets of size $3$ as $P_1^1, \cdots, P_4^1$. Next, the user constructs $P_1^1=\{W_1,{\mathcal{S}}\}=\{1,2,3\}$ and randomly partitions the set of remaining messages into $P_2^1, P_3^1, P_4^1$ sets. Assume the user has chosen $P_2^1=\{4,5,6\}$, $P_3^1=\{7,8,9\}$, $P_4^1=\{10,11,12\}$. Then, the user sends to the server a random permutation of $\{{P_1^1},\cdots, P_4^1\}$. The server picks an arbitrary Cauchy matrix $C=[c_{ij}]_{12 \times 5}$ with parameters over $\mathbb{F}_{17}$ as follows:\vspace{-0.5cm}

\[
C=
  \begin{bmatrix}
    7 & 13 & 6 & 9 & 1 \\
    3 & 7 & 13 & 6 & 9 \\
    5 & 3 & 7 & 13 & 6 \\
    15 & 5 & 3 & 7 & 13 \\
    2 & 15 & 5 & 3 & 7 \\
    12 & 2 & 15 & 5 & 3 \\
    14 & 12 & 2 & 15 & 5 \\
    10 & 14 & 12 & 2 & 15 \\
    4 & 10 & 14 & 12 & 2 \\
    11 & 4 & 10 & 14 & 12 \\
    8 & 11 & 4 & 10 & 14 \\
    16 & 8 & 11 & 4 & 10 \\
  \end{bmatrix}
\]
The server sends back to the user four coded packets as:\vspace{0.125cm}\\ $Y_1=7X_1+3X_2+5X_3$, \\
 $Y_2=15X_4+2X_5+12X_6$, \\
 $Y_3=14X_7+10X_8+4X_9$,\\ 
 $Y_4=11X_{10}+8X_{11}+16X_{12}$,\vspace{0.125cm}\\ 
The user can retrieve $X_1$ by replacing the values of $X_2$ and $X_3$ in $Y_1$. From the server's perspective, the user's demand is in one of the four partitions $\{P_1^1, \cdots, P_4^1\}$ with probability $\frac{1}{4}$, and in each partition, each of the indices is the user's demand index with probability $\frac{1}{3}$. Thus, the probability that each of the indices $i \in [12]$ being as a demand index would be the same, i.e.,  $P(\boldsymbol{W}_1=i|Q^{[W_1=1,\{2,3\}]})=\frac{1}{12}=P_{W_1}(i)$.\vspace{0.1cm}
 
\textbf{Second round:} The user demands the message $X_4$, i.e., $W_2=4$. Based on the OP protocol, the user labels $2$ sets of size $6$ as $P_1^2, P_2^2$. Since $W_2=4 \in P_2^1$ and ${\mathcal{S}} \in P_1^1$, the user constructs $P_1^2=P_1^1 \cup P_2^1=\{1,2,3,4,5,6\}$. For constructing $P_2^2$, the user chooses the remaining two partitions of round 1, i.e., $P_3^1$ and $P_4^1$, and unions them, i.e., ${P_2^2=P_3^1 \cup P_4^1=\{7,8,9,10,11,12\}}$. Then, the user sends to the server a random permutation of $\{{P_1^2}, {P_2^2}\}$. Based on the extracted submatrix $H$ of the selected Cauchy matrix $C$, the server constructs 2 linearly independent combinations of the messages with indices in each partition $\{{P_1^2}, {P_2^2}\}$, and sends back to the user four coded packets as follows:\vspace{0.125cm}\\
$Z_1=13X_1+7X_2+3X_3+5X_4+15X_5+2X_6$,\\
$Z_2=6X_1+13X_2+7X_3+3X_4+5X_5+15X_6$,\\
$Z_3=12X_7+14X_8+10X_9+4X_{10}+11X_{11}+8X_{12}$,\\
$Z_4=2X_7+12X_8+14X_9+10X_{10}+4X_{11}+11X_{12}$.\vspace{0.125cm}\\
 The user has already downloaded $X_1$ from the first round. Thus, from the answers of the first and second rounds, the user can retrieve $X_4$ by solving a set of three linearly independent equations with three unknown as follows:\vspace{0.125cm}\\
 $15X_4+2X_5+12X_6=Y_2$,\\ 
 $5X_4+15X_5+2X_6=Z_1-13X_1-7X_2-3X_3$, \\ 
 $3X_4+5X_5+15X_6=Z_2-6X_1-13X_2-7X_3$.\vspace{0.125cm} \\ 
 From the server's perspective, given the queries $Q^{[W_{1:2},{\mathcal{S}}]}$ and all the packets, the probability that each of the indices $i \in [12]$ being as a user's demand index in the second round would be the same and can be calculated as follows:\vspace{-0.8cm}
\begin{align*}
&\mathbb{P}(\boldsymbol{W}_2= i| Q^{[W_{1:2},{\mathcal{S}}]}, X) 
\\ &= \sum_{{\mathcal{S}}^\star} (\mathbb{P}(\boldsymbol{W}_2= i| Q^{[W_{1:2},{\mathcal{S}}]}, X, {\mathcal{S}}^\star)  \times \mathbb{P}( {\mathcal{S}}^\star|Q^{[W_{1:2},{\mathcal{S}}]}, X)) \\ &=3 \times \frac{1}{3} \times \frac{1}{\binom{3}{2}} \times \frac{1}{4}=\frac{1}{12}.
\end{align*}
Also, from the server's perspective, given the queries $Q^{[W_{1:2},{\mathcal{S}}]}$ and all the packets, the probability that each of the indices $i \in [12]$ being as a user's demand index in the first round is the same, i.e., $\mathbb{P}(\boldsymbol{W}_1= i| Q^{[W_{1:2},{\mathcal{S}}]}, X)= \frac{1}{12}$. This means that the scheme is individually private and from the server's perspective, the user's demand index in the first round will remain private after round 1.\vspace{0.1cm}

\textbf{Third round:} The user demands the message $X_7$, i.e., $W_3=7$. Based on the OP protocol, the user labels $1$ set of size $12$ as $P_1^3$. Since $W_3=7 \in P_2^2$ and ${\mathcal{S}} \in P_1^2$, the user constructs $P_1^3=P_1^2 \cup P_2^2=\{1,2,\cdots,12\}$. Based on the selected Cauchy matrix $C$, the server constructs 2 linearly independent combinations of all the messages, and sends back to the user the following two coded packets:\vspace{0.125cm}\\
\vspace{0.125cm}
$T_1=\sum_{i=1}^{12} c_{i4}X_i$,\\
\vspace{0.125cm}
$T_2=\sum_{i=1}^{12} c_{i5}X_i$.\vspace{0.125cm}\\
 The user has already downloaded $X_1$ from the first round and $X_4$, $X_5$, and $X_6$ from the second round. Thus, from the answers of the first, second and third rounds, the user can retrieve $X_7$ by solving a set of six linearly independent equations with six unknown as follows:\vspace{0.125cm}\\
 \vspace{0.125cm}
 $14X_7+10X_8+4X_9=Y_3$,\\
 \vspace{0.125cm}
 $11X_{10}+8X_{11}+16X_{12}=Y_4$,\\
 \vspace{0.125cm}
 $12X_7+14X_8+10X_9+4X_{10}+11X_{11}+8X_{12}=Z_3$,\\
 \vspace{0.125cm}
 $2X_7+12X_8+14X_9+10X_{10}+4X_{11}+11X_{12}=Z_4$,\\
 \vspace{0.125cm}
 $\sum_{i=7}^{12} c_{i4}X_i=T_1-\sum_{i=1}^6 c_{i4}X_i$,\\
 \vspace{0.125cm}
$\sum_{i=7}^{12} c_{i5}X_i=T_2-\sum_{i=1}^6 c_{i5}X_i$.\vspace{0.125cm}\\
From the server's perspective, given the queries $Q^{[W_{1:3},{\mathcal{S}}]}$ and all the packets, the probability that each of the indices $i \in [12]$ being as a user's demand index would be the same and can be calculated as follows:

\begin{align*}
&\mathbb{P}(\boldsymbol{W}_3= i| Q^{[W_{1:3},{\mathcal{S}}]}, X) 
\\ &= \sum_{{\mathcal{S}}^\star} (\mathbb{P}(\boldsymbol{W}_3= i| Q^{[W_{1:3},{\mathcal{S}}]}, X, {\mathcal{S}}^\star)  \times \mathbb{P}( {\mathcal{S}}^\star|Q^{[W_{1:3},{\mathcal{S}}]}, X)) \\ &=6 \times \frac{1}{6} \times \frac{1}{\binom{3}{2}} \times \frac{1}{4}= \frac{1}{12}.
\end{align*}

Also, from the server's perspective, in the second round the demand of the first round and the demand of the second round remains private, i.e., $\mathbb{P}(\boldsymbol{W}_1= i| Q^{[W_{1:3},{\mathcal{S}}]}, X)=\frac{1}{12}$ and $\mathbb{P}(\boldsymbol{W}_2= i| Q^{[W_{1:3},{\mathcal{S}}]}, X)=\frac{1}{12}$.

\end{example} \vspace{1cm}

\appendix[Proof of Lemma~\ref{privac & recov Conds 1}]
  If there does not exist any ${\mathcal{S}}^\star$ such that ${{X}}_{W^{\star}}$ is recoverable from $A^{[W_{1:i},{\mathcal{S}}]}$ and ${\mathcal{X}}_{{\mathcal{S}}^\star}$, then the server knows that $W^{\star}$ cannot be the user's demand index, and this violates the privacy condition. Given the optimal scheme in the first round, if there exists a ${\mathcal{S}}^\star \nsubseteq P_j$ for some $j \in [n]$, such that ${{X}}_{W^{\star}}$ is recoverable from $A^{[W_{1:i},{\mathcal{S}}]}$, $Q^{[W_{1:i},{\mathcal{S}}]}$ and ${\mathcal{X}}_{{\mathcal{S}}^\star}$, then the server knows that ${\mathcal{S}}^{\star}$ cannot be the user's side information index set. Thus, $W^{\star}$ cannot be the user's demand index, and this violates the privacy condition. 

\bibliographystyle{IEEEtran}
\bibliography{QGTRefs}

\begin{thebibliography}{10}
\providecommand{\url}[1]{#1}
\csname url@samestyle\endcsname
\providecommand{\newblock}{\relax}
\providecommand{\bibinfo}[2]{#2}
\providecommand{\BIBentrySTDinterwordspacing}{\spaceskip=0pt\relax}
\providecommand{\BIBentryALTinterwordstretchfactor}{4}
\providecommand{\BIBentryALTinterwordspacing}{\spaceskip=\fontdimen2\font plus
\BIBentryALTinterwordstretchfactor\fontdimen3\font minus
  \fontdimen4\font\relax}
\providecommand{\BIBforeignlanguage}[2]{{%
\expandafter\ifx\csname l@#1\endcsname\relax
\typeout{** WARNING: IEEEtran.bst: No hyphenation pattern has been}%
\typeout{** loaded for the language `#1'. Using the pattern for}%
\typeout{** the default language instead.}%
\else
\language=\csname l@#1\endcsname
\fi
#2}}
\providecommand{\BIBdecl}{\relax}
\BIBdecl

\bibitem{chor1995private}
B.~Chor, O.~Goldreich, E.~Kushilevitz, and M.~Sudan, ``Private information
  retrieval,'' in \emph{Foundations of Computer Science, 1995. Proceedings.,
  36th Annual Symposium on}.\hskip 1em plus 0.5em minus 0.4em\relax IEEE, 1995,
  pp. 41--50.

\bibitem{kadhe}
S.~Kadhe, B.~Garcia, A.~Heidarzadeh, S.~El~Rouayheb, and A.~Sprintson,
  ``Private information retrieval with side information: The single server
  case,'' in \emph{Communication, Control, and Computing (Allerton), 2017 55th
  Annual Allerton Conference on}.\hskip 1em plus 0.5em minus 0.4em\relax IEEE,
  2017, pp. 1099--1106.

\bibitem{heidarzadeh2018capacity}
A.~Heidarzadeh, F.~Kazemi, and A.~Sprintson, ``Capacity of single-server
  single-message private information retrieval with coded side information,''
  \emph{arXiv preprint arXiv:1806.00661}, 2018.

\bibitem{tandon2017capacity}
R.~Tandon, ``The capacity of cache aided private information retrieval,'' in
  \emph{Communication, Control, and Computing (Allerton), 2017 55th Annual
  Allerton Conference on}.\hskip 1em plus 0.5em minus 0.4em\relax IEEE, 2017,
  pp. 1078--1082.

\bibitem{wei2017fundamental}
Y.-P. Wei, K.~Banawan, and S.~Ulukus, ``Fundamental limits of cache-aided
  private information retrieval with unknown and uncoded prefetching,''
  \emph{arXiv preprint arXiv:1709.01056}, 2017.

\bibitem{wei2018cache}
------, ``Cache-aided private information retrieval with partially known
  uncoded prefetching: Fundamental limits,'' \emph{IEEE Journal on Selected
  Areas in Communications}, 2018.

\bibitem{chen2017capacity}
Z.~Chen, Z.~Wang, and S.~Jafar, ``The capacity of private information retrieval
  with private side information,'' \emph{arXiv preprint arXiv:1709.03022},
  2017.

\bibitem{shariatpanahi2018multi}
S.~P. Shariatpanahi, M.~J. Siavoshani, and M.~A. Maddah-Ali, ``Multi-message
  private information retrieval with private side information,'' \emph{arXiv
  preprint arXiv:1805.11892}, 2018.

\bibitem{sun2016capacity}
H.~Sun and S.~A. Jafar, ``The capacity of private information retrieval,'' in
  \emph{Global Communications Conference (GLOBECOM), 2016 IEEE}.\hskip 1em plus
  0.5em minus 0.4em\relax IEEE, 2016, pp. 1--6.

\bibitem{sun2018capacity}
------, ``The capacity of robust private information retrieval with colluding
  databases,'' \emph{IEEE Transactions on Information Theory}, vol.~64, no.~4,
  pp. 2361--2370, 2018.

\bibitem{banawan2018capacity}
K.~Banawan and S.~Ulukus, ``The capacity of private information retrieval from
  coded databases,'' \emph{IEEE Transactions on Information Theory}, vol.~64,
  no.~3, pp. 1945--1956, 2018.

\bibitem{banawan2018multi}
------, ``Multi-message private information retrieval: Capacity results and
  near-optimal schemes,'' \emph{IEEE Transactions on Information Theory}, 2018.

\bibitem{heidarzadeh}
A.~Heidarzadeh, B.~Garcia, S.~Kadhe, S.~E. Rouayheb, and A.~Sprintson, ``On the
  capacity of single-server multi-message private information retrieval with
  side information,'' \emph{arXiv preprint arXiv:1807.09908}, 2018.

\end{thebibliography}

\end{document}